\documentclass{article}
\usepackage{mathptmx} 
\usepackage{amsmath}
\usepackage{amsfonts}
\usepackage{amssymb}
\usepackage{datetime}
\usepackage{geometry}
\newtheorem{theorem}{Theorem}

\newtheorem{corollary}{Corollary}
\newcommand{\proof}{Proof:\ }
\newcommand{\qed}{$\Box$}
\numberwithin{equation}{section}
\usepackage[absolute]{textpos}
\usepackage{setspace}
\usepackage{hyperref}
\usepackage{graphics}

\title{The Swapped Dragonfly}
\begin{document}
\date{}
\maketitle
\begin{centering}
	Richard Draper\footnote{Center for Computing Sciences, Institute for Defense Analyses, Bowie MD}
		
	rndrape@super.org,rndrape@gmail.com
	
\end{centering}
\thispagestyle{empty}
\vspace{3em}

\bibliographystyle{abbrv}
\large{

\begin{abstract}
	This paper describes the Swapped Dragonfly. It is a two-parameter family of diameter three interconnection networks, $D3(K,M)$, which are linearly scalable in $M$. Although $D3(K,M)$ is a Dragonfly, it differs from standard Dragonflies in many respects. It has a $K\times M\times M$ coordinate system $(c,d,p)$. The routers $(c,d,p)$ and $(c',p,d)$ are globally connected using a swap of $p$ and $d$.
	
	If $L<K$ and/or $N<M$, $D3(K,M)$ contains $D3(L,N)$. The coordinate system enables source vector routing on $D3(K,M)$. A source-vector induces $KM^2$ parallel paths on $D3(K,M)$. Because of this, the Swapped Dragonfly can support conflict-free parallelism over local ports, global ports, routers and source-vectors. In particular, there is an all-to-all algorithm which is not a pairwise exchange algorithm.
\end{abstract}

\section{INTRODUCTION}

The interconnection network is an important and costly component of a modern supercomputer. It is the backbone of the system, potentially limiting or enhancing performance on parallel applications. An important example is the Dragonfly network. It is the topology of the IBM PERC \cite{Arim:Inter} and a variation of which is the network of the $CRAY XC^\copyright$ \cite{alv:CRAY} family of computers. This paper defines the Swapped Dragonfly, denoted $D3(K,M)$, and discusses its properties. The network is an interconnection network in the Dragonfly family of networks.

The objective of the research leading to this paper was to design networks that were linearly scalable of low diameter, and were useful in performing communicative primitives that occur frequently in scientific applications. These primitives are broadcast, one-to all, all-to-one, and all-to-all. The Swapped Dragonfly has diameter three and is linearly scalable. The communication primitives are developed in this paper.

Linear scaling is important for the following reason. The purchase of a large supercomputer follows a lengthy trajectory involving, among other things, planning, budgeting, contracting, delivery, and testing. Typically, the machine consists of cabinets with removable drawers which themselves are computers which (sometimes) can be upgraded. The cabinets are connected by cables. The purchasing process can take a long time during which costs and budgets change. Budgets seldom increase but costs always increase. Less money and higher costs translate into fewer cabinets. Having an interconnection network which supports machines of size $M^2, 2M^2,3M^2,\dots$ makes increasing or decreasing the size of the machine easy. This is linear scaling. 
 
The Swapped Dragonfly is organized in $K$ cabinets containing $M$ drawers each containing $M$ routers. It is denoted $D3(K,M)$. Routers have $K$ global and $M-1$ local bidirectional ports. Routers also have ports to which compute nodes are attached. The address space is $(c,d,p) \ c\bmod K$ and $(d,p) \bmod M$. The coordinate $c$ identifies a set of $M^2$ routers called a cabinet and $(c,d)$ identifies a drawer. There are two networks: a local network connecting the routers in a drawer in a complete graph and a global network connecting the drawers. $D3(K,M)$ is a Dragonfly because it is based on a building block which is a complete graph and it scales linearly. It is a swapped network because the global network is defined by a swap, $(c,d,p) \leftrightarrow (c',p,d)$. 

$D3(k,M)$ is treated as a packet switching network. Packet headers contain a counter and a source-vector. The longest path has three hops so a source-vector consists of three ports $(\gamma,\pi,\delta)$. This represents a path for every router in the network. On $D3(K,M)$ these paths are parallel. Source-vectors are used to design algorithms for broadcast, one-to-all, all-to-one, and all-to-all. Because of the parallelism of the source-vector paths, there are no inter-round conflicts in any of these algorithms. The counter is used to remove intra-round conflicts so the resulting algorithms are free of link conflicts. These algorithms exploit parallelism over local ports, global ports, routers, and source vectors. The all-to-all algorithm is unusual because it is not hierarchical and it is not a pairwise exchange algorithm. 

The Swapped Dragonfly network is intended to be visible to the user and algorithms are tailored to the topology. In fact, the network is part of the algorithm. This is analogous to Cannon's matrix multiply algorithm on a mesh architecture \cite{cann:alg}; all-to-all exchange on a hypercube; and ascend-descend algorithms on a hypercube and a cube-connected cycle \cite{prp:CCC}.

There are other useful properties of $D3(K,M)$. The cabinets of $D3(K,M)$ in a subset of $\{0,\cdots,K-1\}$ of size $J$ are connected in a subnetwork isomorphic to $D3(J,M)$. As a result, not only is $D3(K,M)$ linearly scalable in $K$ but $D3(J,M)$ can be expanded to $D3(K,M)$ without altering the original network of $D3(J,M)$. Although the swap may seem to introduce confusion in the wiring diagrams of the network, the wiring is well-defined and there is a simple and intuitive way to describe the wiring in terms of the cabinet structure of the network. 

The paper is organized as follows: Sections two through six define $D3(K,M)$ and discuss physical implementation and other properties useful in an interconnection topology. The next two sections present capabilities assumed for the routers and the synchronized source-vector routing algorithm. Section nine presents the results on communication primitives and identifies properties of the network and the routers required to achieve the results claimed. The proofs of these claims are simple but tedious. However, the reader should examine the algorithms to understand where the parallelism is. Proofs are presented in an Appendix. Section ten converts the source-vector packet headers to destination headers and uses these headers to enable deflection routing. Section eleven compares $D3(K,M)$ to a standard Dragonfly. The last three sections are a discussion of related work, a conclusion, and a bibliography.

\section{THE NETWORK D3(K,M)}

The network has $KM^2$ routers, each having $K$ global and $M-1$ local ports. We denote the network $D3(K,M)$.
Routers are addressed by
\begin{equation*}
(c ,d ,p )\ \ \textrm{with}\ \   d\ \textrm{and}\ p\ \textrm{mod}\ M,\ c\ \textrm{mod}\  K.
\end{equation*} 
The coordinates are referred to as (cabinet, drawer, router).
Connectivity is defined by
\begin{equation}\label{eq:connect}
 (c,d,p)\stackrel{l}{\longleftrightarrow} (c,d,p')\ \ \textrm{and}\  
(c,d,p)\stackrel{g}{\longleftrightarrow} (c',p,d).
\end{equation} 
The $l$ and $g$ are referred to as \emph{local} and \emph{global} communications, respectively. \emph{Note the swap of $p$ and $d$ in the global communication}. It is assumed that all links are bidirectional and that $g$ and $l$ communications can occur simultaneously. \emph{It is assumed that the $M$ routers in a drawer are connected locally in a complete graph.} Each router has K global ports denoted  $a \bmod K$. They are included in expanded notation $(c,d,p,(a ))$. Each router has $M-1$ local ports. In the algorithms of this paper, local ports are indexed $1,\cdots,M-1$. Routers are labeled from $0,\cdots,M-1$. Local port $\pi$ on router $p$ connects to local port $-\pi$ on router $p+\pi$. Arithmetic is modulo $M$. There is no local port $0$. Its reference means that the packet being routed does not move during the time in question. Three figures demonstrating the structure of the Swapped Dragonfly $D3(3,4)$ appear at the end of the Appendix.
It follows from (\ref{eq:connect}) that connecting  $(c,d,p)$ to $(c',d',p')$ may be done with one global and  two  local hops: 
\begin{equation} \label{eq:lgl} 
\begin{array}{crcl}
step & & path & \\
1 & (c,d,p) & \stackrel{l}{\longleftrightarrow} & (c,d,d') \\
2 & (c,d,d',(c'-c)) & \stackrel{g}{\longleftrightarrow} & (c',d',d,(c-c')) \\ 
3 & (c',d',d ) & \stackrel{l}{\longleftrightarrow} &  (c',d',p').
\end{array}
\end{equation}

If $(c',d',p') \neq (c,d,p)$ and $d' \neq p$ and $p' \neq d$, (\ref{eq:lgl}) is the shortest path and it is unique. If $d=d'$ there is a $glg$ path as well as an $lgl$ path.
There are $K$ four hop paths of the form $glgl$ and $K$ paths of the form $lglg$ that connect $(c,d,p)$ to $(c',d',p')$. There are also paths which do not fit the $lgl$ pattern.
For example, $(c,d,p)$ is connected to $(c',p,d)$ by the $gg$ path $(c,d,p)\stackrel{c'-c}{\longleftrightarrow}(c',d,p)\stackrel{0}{\longleftrightarrow}(c',d,p)$.

\section{PHYSICAL IMPLEMENTATION}
The global network has features that facilitate its  physical implementation. 
All drawers are identical. Each cabinet is a $D3(1,M)$ so all cabinets are identical. It follows from step 2 of ( \ref{eq:lgl}) that
$$ (c,d,p,(\gamma)) \stackrel{g}{ \longleftrightarrow}(c+\gamma,p,d, (-\gamma)), \textrm{ for } p = 0 ,\dots, M-1.$$

(The minus signs on the port addresses result from the bidirectional property of global communications.) $(c,d,*,(\gamma))$ denotes global port $\gamma$ on every router in drawer $(c,d)$. $(c,*,p,(\gamma))$ denotes global port $\gamma$ of router $p$ in every drawer of cabinet $c$. $(c,d,*,(\gamma))$ connects in order to $(c+\gamma,*,d,(-\gamma))$. For example, if $K=6,\  (4,5,3,(4))$ connects to $(2,3,5,(2))$ because $4+4=2 \bmod 6$ and $-4=2 \bmod 6$. If $K$-wide ribbons were used for global connections one end would attach to a column $(c + \gamma,*,d,(-\gamma))$ and the other end would attach to $(c,d,*,(\gamma))$.

\section{SUBNETWORKS AND PARTITIONS}
If $\kappa$ is a subset of $\{0,\cdots,N-1\}$ of size $K$, it is clear that the cabinets of $D3(N,M)$ with $c \in \kappa$ look like $D3(K,M)$. If $\lambda$ is a subset of $\{0,\cdots,M-1\}$ of size $L$, then $\{(c,d,p) \in D3(N,M)\, \vert\, \{d,p\} \in \lambda \time \lambda\}$ is closed undet the action of global links. Therefore, this set looks like $D3(N,L)$. Combining both constraints yields a network that looks like $D3(K,L)$ inside $D3(N,M)$.

For this to make sense and to enable the translation of a source-vector algorithm on $D3(K,L)$ to an algorithm on the associated subspace of $D3(N,M)$ requires two sets of tables. The first set associates $\gamma \in \{0,\cdots,K-1\}$ with $\gamma' \in \kappa$ at $c\in\kappa$. The second set associates $\pi \in \{0,\cdots L-1\}$ with $\pi' \in \lambda$ at $p\in \lambda.$ Requiring tables at each router to translate a vector $(\gamma,\pi,\delta)$ on $D3(K,L)$ to a vector $(\gamma',\pi',\delta')$ is not a heavy burden. This is done at a host attached to the router. The following theorem gives the details for $c$ and $\gamma$, and is followed by an example.

\begin{theorem}\label{thm:additive part}
	Assume routers have $N$ global ports. Let $K < N$ and let
	$\kappa = \{k_0, k_1, \dots, k_{K-1}\}$ be a  subset of $\{0, \dots, N-1\}$.
	Let $a_{j,i} = k_j - k_i\!\! \mod N$ and
	$$D3(\kappa,M,N) = \{(k_i,d,p,(a_{j,i}))\ | \ 0 \leq i < K, 0 \leq j < K,\forall (d,p)\}.$$\ 
	$D3(\kappa,M,N)$, a subnetwork of $D3(N,M)$, is isomorphic to $D3(K,M)$.
\end{theorem}
\proof
Associate $i$ with $k_i$ and $j-i\!\! \mod K$ with $a_{j,i}$. Then 
\begin{align*}
(k_i,d,p,(a_{j,i})) &\stackrel{g}{\longleftrightarrow} (k_j,p,d,(a_{i,j}))\\
\intertext{corresponds to}
(i,d,p,(j-i\!\! \mod K)) &\stackrel{g}{\longleftrightarrow} (j,p,d,(i-j\!\! \mod K)).
\end{align*} 
An analogous proof applies when $d$ and $p$ are restricted to a subset of $\{0,\dots,M-1\}$.
\qed
\newline
It follows from this theorem that $D3(N,M)$ contains non-intersecting subnetworks, $D3(K_i,M) $, based upon any partition of the set 
$\{0, \dots, N-1\}$. This property also makes it possible to scale a network with $K$ cabinets up to  any size less than or equal to $N$ at only the 
added cost of the additional cabinets, their global wiring, and their global wiring into the original network. No existing wiring need be moved provided that the original network was wired as $D3(\kappa,M,N)$. 

The table below identifies the cabinets and global ports used to create a $D3(4,M)$ and a complementary $D3(5,M)$ inside $D3(9,M)$.

\begin{equation*}
\begin{array}{cccccccc}
i& k_i & ports &&&i & k_i& ports \\
0&1&\{0,1,4,7\}&&&0&0&\{0,3,4,6,7\}\\
1&2&\{3,0,6,8\}&&&1&3&\{8,0,1,3,4\}\\
2&5&\{5,6,0,3\}&&&2&4&\{5,8,0,2,3\}\\
3&8&\{3,6,5,0\}&&&3&6&\{3,6,7,0,1\}\\
&&&&&4&7&\{2,5,6,8,0\}
\end{array}
\end{equation*}
Each cabinet $i$ of $D3(4,M)$ is assigned to $k_i$ on $D3(9,M)$. The vector of ports is assigned to each cabinet. To convert a $D3(4,M)$ vector $(\gamma,\pi,\delta)$ at $(i,d,p)$ to a $D3(9,M)$ vector $(\gamma',\pi,\delta)$ at $(k_i,d,p)$,$\gamma'$ is the element $\gamma$ of vector $i$.
Both local and global subnetworks can be used to create isolated subnetworks of type $D3$. Equally important, both subnetworks  can be used for maintenance. If drawer $(c,d,p)$ needs to be replaced, then maintenance can be performed while the network  $D3(K,M-1)$ is operating. This takes $K(2M-1)$ routers off-line but leaves a working network with topology in the same family as the full network.   In the second case a cabinet can be taken off line leaving a $D3(K-1,M)$ running. This takes $M^2$ routers off line.

\begin{corollary}
	$D3(K,M)$ has a cutset of size $\min(K^2M^2/2,KM^3/2).$\\
\end{corollary}	
\proof
Partition $D3(K,M)$ into $D3(K/2,M)$. Every router is denied access to $K/2$ global ports. This gives a cutset of size $K^2M^2/2$. Partitioning into $D3(K,M/2)$ yields the other value.
\qed

\section{GLOBAL LINK CONFLICTS}
Transparency of link conflicts and an upper bound on their cost are two of the features of these networks. 

\begin{theorem}\label{thm:conflict}
	Given simultaneous transmissions
	$$
	(c,d,p)\rightarrow (c',d',p') \ \textrm{and}\ (\gamma,\delta,\pi) \rightarrow (\gamma',\delta',\pi')
	$$
	with $(c',d',p') \neq (\gamma',\delta',\pi')$,
	link conflict occurs on minimal paths if  and only if $(c,d) = (\gamma,\delta)$ and  $(c',d') = (\gamma',\delta')$. 
\end{theorem}
\begin{proof}   
	We apply \ref{eq:lgl} to two routers:
	\begin{equation*}
		\begin{split}
			&(c,d,p) \stackrel{l}{\rightarrow} (c,d,d')
			\stackrel{g}{\rightarrow} (c',d',d,)
			\stackrel{l}{\rightarrow} (c',d',p').\\ 
			&(\gamma,\delta,\pi) \stackrel{l}{\rightarrow}(\gamma,\delta,\delta') 
			\stackrel{g}{\rightarrow} (\gamma',\delta',\delta)
			\stackrel{l}{\rightarrow}(\gamma',\delta',\pi') 
		\end{split}
	\end{equation*}
	In each line, four routers occur.  The first and fourth routers are distinct.  
	The second routers are identical if and only if
	\begin{equation}
		c- \gamma = d- \delta = d' - \delta' =0 \label{eq:sec}
	\end{equation}
	and the third routers are identical if and only if
	\begin{equation}
		c' - \gamma' = d' -\delta' = d-\delta=0.\label{eq:third}
	\end{equation}

	If only \ref{eq:sec} holds, then the communication is
	\begin{equation*}
		\begin{split}
			&(c,d,p) \stackrel{l}{\rightarrow} (c,d,d')
			\stackrel{g}{\rightarrow} (c',d',d)
			\stackrel{l}{\rightarrow} (c',d',p')\\
			&(c,d,\pi) \stackrel{l}{\rightarrow} (c,d,d')
			\stackrel{g}{\rightarrow} (\gamma',d',d)
			\stackrel{l}{\rightarrow} (\gamma',d',\pi').\\
		\end{split}
	\end{equation*}
	This router conflict is not a problem because router $(c,d,d')$ is receiving packets over distinct local ports
	and sending packets out over distinct global ports. If only \ref{eq:third} holds, then the communication is
	\begin{equation*}
		\begin{split}
			&(c,d,p) \stackrel{l}{\rightarrow} (c,d,d')
			\stackrel{g}{\rightarrow} (c',d',d,)
			\stackrel{l}{\rightarrow} (c',d',p')\\
			&(\gamma,d,\pi) \stackrel{l}{\rightarrow} (\gamma,d,d')
			\stackrel{g}{\rightarrow} (c',d',d,)
			\stackrel{l}{\rightarrow} (c',d',\pi').\\
		\end{split}
	\end{equation*}
	Analogous to the previous case, $(c',d',d)$ is receiving packets over distinct global ports
	and sending packets out over distinct local ports.

	If  both \ref{eq:sec} and \ref{eq:third}  hold then the communication is
	\begin{equation}
		\begin{split}
			&(c,d,p) \stackrel{l}{\rightarrow} (c,d,d')
			\stackrel{g}{\rightarrow} (c',d',d,)
			\stackrel{l}{\rightarrow} (c',d',p')\\
			&(c,d,\pi) \stackrel{l}{\rightarrow} (c,d,d') \label{eq:both}
			\stackrel{g}{\rightarrow} (c',d',d,)
			\stackrel{l}{\rightarrow} (c',d',\pi').
		\end{split}
	\end{equation}
	Router $(c,d,d')$ has to send two packets out over the same port.
\end{proof} 

The conflict  has a physical interpretation.
Two routers on one drawer are sending packets to two routers on another drawer.
The global communication in \ref{eq:both}  has two messages to send out over one and the same port.
In the worst case, it could have $M$ messages to send out over one port. In programing a loop over the address parameters of the routers, 
having pairs of drawers simultaneously exchange packets could be a  natural thing to do. On the other hand, knowing the source of link contention can be used to mitigate its effect. 

\section{LINEAR AND QUADRATIC SCALING}
This system scales linearly in the number of cabinets which is desirable for several reasons. Manufacturers of massively parallel computers start with a router with a fixed number of ports. Product lines will use this router for several years \cite{alv:CRAY}. It is desirable for the product line to have several size machines. The machines are often built to order. Linear scaling means that the product line can contain an arithmetic progression of sizes up to the maximum possible when using all of the ports on a router. Large machines are typically delivered over a significant period of time. Users work on the 
early deliveries to understand computational power and programming issues. A two cabinet $D3(2,M)$ machine provides an accurate programming image of the $K$-cabinet $D3(K,M)$. If budgets scale at all, it is linearly.  More realistically, budgets are often cut or cost estimates exceeded before a machine is delivered. Reducing the value of $K$ is a graceful way of dealing with 
these realities.

The term \emph{scalable} can be interpreted in another way. The second interpretation is
that routers with $N$ global ports are used to build a system with $K<N$ cabinets which subsequently may be
expanded to a system with  $N$ cabinets. A system with $K<N$ leaves the 
extra global ports unused. One would like the ability to build a machine out of uniform cabinets, routers, and drawers which can be expanded
with minimal wasted cost and effort to a larger machine. To put the bigger machine together, some of the global network wiring in the original
machine will have to be changed unless the original machine was wired as $D3(\kappa,M,N)$. See the discussion of the example in Section 4. That being the case, the expanded machine will not require rewiring the existing machine. This is called \emph{graceful} scaling.

The system also scales quadratically.  If each drawer has $M$ slots for routers, but contains only $N<M$ routers then the system $D3(K,N)$ can be built. The routers must occupy the same numerical position in every 
drawer.  The resulting system has $KN^2$ routers. This scaling calls for a fixed number of cabinets which are only partially 
populated. Like linear scaling, adding routers to the drawers can be done without disturbing existing cables.

\section{ROUTER CAPABILITIES}

A router has $K$ global ports and $M-1$ local ports\footnote{Modular arithmetic is used to design the algorithm of this paper. Consequently, the address space for local ports includes zero, so reference maybe made to $M$ local ports.}. It is desirable that the local and global networks are balanced. Letting $B$ denote $\max\{K,M\}$, balance means that the a router can simultaneously send and receive $B/K $  messages over each of  its global ports and $B/M$ messages over each of its local ports.  Consequently, a router can send and receive $2B$ messages. We refer to the time to do this as a \emph{time step}. However, it is very difficult to prove anything about performance on communication primitives if this form of balance is assumed. Therefore, it will be assumed that \emph{the time to send a message over a global link is the same as the time to  send a message over a local link}. This is a \emph{time step}. With this definition of balance \emph{we assume that a router can simultaneously send and receive $K$ messages over its global ports and $M-1$ messages over its local ports in one time step.} It is also assumed that a router can broadcast a packet out of all its ports in one time step. Modern routers can do this.

Sending a packet through a router and across an attached link is called a \emph{hop} or a network hop. The performance of a one round algorithm is measured in hops. The performance of a multi-round algorithm is measured in number of rounds and intra-round delays. The sum of these two numbers is within a few (startup) hops of the number of hops required to run the algorithm provided that the rounds can be pipelined. 

\section{ROUTING}

Routing can be implemented by a packet header consisting of the destination address. To transmit a packet, the difference between 
sending and receiving addresses must be calculated. 
Given that information, the sending router must decide what to do with the packet. The procedure is usually conducted by table look-up. This routing  has the advantage that a packet which is detoured because of a congested or broken link can wander about 
and eventually reach its destination.

This paper uses source-vector routing with a synchronizing counter. The packet header has fields $(b;\gamma,\pi,\delta)$; $b$ is a counter, $\gamma$ is a global port and $\pi \textrm{ and } \delta$ are local ports. The value of $b$ determines which of the three ports are used. The evolution of a path is:
\begin{equation*}
\begin{array}{rcc}
& router &  header\\ 
& (c,d,p) &  (3;\gamma,\pi,\delta)\\
\delta &\downarrow & \downarrow\\
& (c,d,p+\delta) & (2;\gamma,\pi ,\delta)\\
\gamma &\downarrow & \downarrow\\
&(c+\gamma,p+\delta,d) & (1;\gamma,\pi,\delta)\\
\pi &\downarrow & \downarrow\\
&(c+\gamma,p+\delta,d+\pi)  & (0;  \gamma,\pi,\delta).
\end{array}
\end{equation*}
The packet is initiated by a node attached to $(c,d,p)$. When $b=0$ the packet has arrived and is passed to a node attached to $(c+\gamma,p+\delta,d+\pi)$. The path from $(c,d,p)$ to $(c',d',p')$ uses header $(3;c'-c,p'-d,d'-p)$. The sending router modifies the counter so that the receiving router knows what to do with the packet.

For example, $(c,d,p)$ sending a packet to itself is implemented by the header $(3;0,p-d,d-p)$ which requires 
three hops to reach $(0;0,0,0)$ even if $p = d$. This header induces a three  step path to stand still. This may seem absurd, but if a permutation is being transmitted by the network, there will be routers which talk to themselves. Taking three hops to do so will synchronize the entire 
permutation. Effectively, the headers create a geometry in which every router is three hops away. 

Synchronized source-vector routing is used to design accelerated algorithms for the communication primitives: broadcast \footnote{A broadcast bit is added to the header in implementing the broadcast algorithm}, one-to-all, all-to-one, and all-to-all. Parallel vector flows across the network are used to pipeline the algorithms. The fact that all paths are of length three is crucial to synchronizing these algorithms. The results are stated in \S9 \ and the proofs are in the Appendix.

In the case that drawer-to-drawer communication cannot be avoided, the resulting conflict can be avoided on $D3(K,M)$. A four hop path $(4;\gamma,\pi,\delta)$ is a $glgl$ path $(\gamma,\pi,0,\delta)$. The initial $\gamma$ converts routers in a drawer to routers in a column. This eliminates drawer-to-drawer conflicts.

A drawer is a complete graph on $M$ nodes so uses only $M-1$ ports. The algorithms in this paper refer to local port $0$. When local port $0$ is used, it means that a packet does not move for one step of the counter. The router must have the ability to hold a packet for that step.

\section{COMMUNICATION PRIMITIVES}
It is customary to measure theoretical performance of an algorithm in time steps. The algorithms of this paper consist of a series of parallel rounds. There are no inter-round link conflicts. If the rounds are pipelined, there are intra-round conflicts. These conflicts are resolved by delaying a round. Therefore, the performance of a pipelined algorithm is proportional to the number of rounds plus the number of delays. This is the total number of time steps. The time to complete a single round is measured in network hops. This and startup time are not recorded here.  

Assume $M$ is even and no less than $4$. The performance of communication primitives on $D3(K,M)$ is listed below. The conflicts are the result of pipelining. Detailed statements of the algorithms and proofs are presented in the Appendix. In a round some or all routers receive one or more packets from an attached node.

\begin{enumerate}
	\item An all-to-all exchange can be performed in $KM^2$ rounds with $KM$ intra-round conflicts.
	\item A router $(c,d,p)$ can perform a one-to-all in 
	\begin{enumerate}
		\item $2KM$ rounds with $M$ intra-round conflicts if $d=p$.
		\item $KM$ rounds if $p \neq d$.
	\end{enumerate}
	\item A router $(c,d,p)$ can perform an all-to-one in $KM$ rounds if $d \ne p$.
	\item A router $(c,d,p)$ can perform a broadcast in three hops. It can perform $N$ broadcasts in $N$ rounds if $d \neq p$ and in $2N$ rounds if $d=p$.\footnote{Broadcast differs from the other primitives because a router cannot duplicate packets. Therefore, the packets have to hop off and on the network in order to be duplicated and assigned the outgoing port. This feature makes it easy to coordinate. Therefore, $N$ rounds are be pipelined in $N+2$ hops.} 
	\item $D3(K,M)$ can perform a permutation in $M$ hops.
\end{enumerate}

All of these algorithms use parallelism and contain no inter-round link conflicts. The algorithms are controlled by the sync counter and/or the broadcast bit. As stated before, the time to cross a link is assumed constant. However, that may be because the sync counter controls the time.The result in $9.1$ uses parallelization over all routers in $D3(K,M)$. The results in $9.2a$ and $9.2b$ use parallelization across local ports. The algorithm in $9.3$ combines successive broadcasts with synchronized responses by $M$ routers at a time. The broadcast requires that a router can pass a message to all of its local or all of its global ports simultaneously when the broadcast bit is set.

Here are the demands placed upon the network, routers and nodes by each of these results. All require that the routers can be synchronized.
\begin{enumerate}
	\item A packet with $\pi$ or $\delta$ equal $0$ can be held in the router for a time-step.
	\item A compute node can launch $M$ packets simultaneously.
	\item A router can receive a packet and send it simultaneously out all $K$ global ports or all $M$ local ports.
\end{enumerate}
For these algorithms to perform in time proportional to the number of rounds, they have to be synchronized and protected from interference by unrelated packets.

\section{HEADERS AND DEFLECTION ROUTING}
The packet headers in \S8 \ are useful for finding parallel flows through the network. That is the key to efficient implementation of the communication primitives in \S9. Source-vector headers can be replaced with headers containing the destination of the packet. Doing this makes it possible to introduce adaptive deflection routing to the system. The new headers contains $(b;destination,location)$ denoted $(b;(c',d',p'),(c,d,p))$. The location is updated after each hop of the path.

Routing can be handled by table lookup. There are two port tables, local and global
$$
\begin{array}{lclcllll}
\begin{array}{lcll}
\textrm{local } \ M \times M && \textrm{global } \ K \times K\\

\ \ \ \ \ \cdots a \cdots&&\ \  \ \ \ \ \cdots a \cdots\\
\vdots&&\vdots\\
b \ \  \  \ \ \ \ \ x&& b \ \ \ \ \ \ \ \ \ y\\
\vdots &&\vdots
\end{array}
\end{array}
$$

With $x=b-a \bmod M \textrm{ and } y=b-a \bmod K$, the row entry is taken from the destination and the column entry is taken from the location. The sync counter controls what is read out of the tables. Given the packet header, the array below gives the table look-up determined by the sync counter.
$$
\begin{array}{lccc}
(3;(c',d',p'),(c,d,p)) & \textrm{ go to } & (d',p) \textrm{ local }\\
(2;(c',d',p'),(c,d,p)) & \textrm{ go to } & (c',c) \textrm{ global }\\
(1;(c',d',p'),(c,d,p)) & \textrm{ go to } & (p',d) \textrm{ local }
\end{array}
$$

The diagonal of both tables is $0$. In the case of the global table there is a global port $0$. But in the case of the local port, there is no $0$ port. Because a drawer is a complete graph on $M$ vertices, the local ports are properly labeled $1$ to $M-1$. A path using the $0$ local port means the packet does not move. Its sync counter must step down. The packet has to be held somewhere in the router for that step before being moved to a global port buffer or the arrival buffer.

Destination headers make it easy to introduce deflection routing.The range of the counter $b$ is increased to $5$. The router has to do more than a table lookup.
$$
\begin{array}{lccc}
(5;(c',d',p'),(c,d,p)) & \textrm{ takes random local port } D\\
(4;(c',d',p'),(c,d,D)) & \textrm{ takes random global port } C
\end{array}
$$
After these two steps the packet is at $(C,D,d)$ and has $b=3$ so goes to its destination. There are $M \ glgl $ paths from $(c,d,p)$ to $(c',d',p')$. This is a version of UGAL-G \cite{kim:dragonfly}. If only $b=4$ were added it would be a version of Valiant deflection routing. $D$ and $C$ need not be random but may be selected based on local conditions. $D$ depends on the state of local ports on router $(c,d,p)$. $C$ depends on the state of the global ports on drawer $(c,d)$. The decision to deflect is made when the packet is launched.

This form of deflection routing is not possible if source-vector routing is used. The problem is that the vector $(\gamma,\pi,\delta)$ takes source $(c,d,p)$ to $(c+\gamma,p+\delta,d+\pi)$. A deflection takes $(c,d,p)$ to $(C,D,d)$. The vector needed to reach the destination is  $(c+\gamma-C,d+\pi-D,p+\delta,d+\pi)$. Routers would have to make the calculation.

The deflection routing technique proposed here provides an additional opportunity if the deflection is managed by an attached node. Given destination $(c',d',p')$ at source $(c,d,p)$, the deflection header $(4;c',d',p')$ with non-random $C=c'-c$ leads to a $glgl$ path: $$(c,d,p) \stackrel{g}{\longrightarrow} (c',p,d) \stackrel{l}{\longrightarrow}(c',d,d') \stackrel{g}{\longrightarrow}(c',d',d) \stackrel{l}{\longrightarrow}(c',d',p')$$ The path starts with a jump to the destination cabinet $c'$. The header contains the destination $(c',d',p')$, and $b=3$ after the first hop. Therefore, the packet stays in cabinet $c'$ and goes to its destination. This is a $glgl$ path which is determined by $(c,d,p) \textrm{ and } (c',d',p')$. This path can be used to design an algorithm parallel over global ports. For example, a one-to-all algorithm can be designed which runs in time proportional to $M^2 \textrm{ if }p \neq d$. These headers can also be used to construct a dilation four embedding of a hypercube of size $2^{k+2m} \leq KM^2$ in $D3(K,M)$ 

The algorithm used for the all-to-all exchange is a program run by a node attached to each router. It exhausts over the space of source vectors. In order to determine the contents of a packet sent along vector $(\gamma,\pi,\delta)$, node $(c,d,p)$ has to compute destination $(c+\gamma,p+\delta,d+\pi)$ of the packet. Therefore, the calculation is equivalent to converting the packet header to a destination header. Using destination routing makes it possible to have both an accelerated all-to-all and deflection routing. It requires that the router contains look-up tables and it may require that the router can choose a $D$ and $C$. This adds to the complexity of the router. However, nodes are attached to drawer $(c,d)$ at router $p$. If the nodes maintain information on traffic on the drawer to which they are attached, then a node can choose $D$ and $C$ at the time a packer is launched. This would keep the router simple.

\section{COMPARISON OF D3(K,M) AND THE DRAGONFLY} 
A Dragonfly network \cite{kim:dragonfly} has groups (drawers) of $M$ routers connected by $M-1$ local ports in a complete graph. Each router has $K$ global ports. A \emph{maximal} Dragonfly, denoted here by $MDF(K,M)$, has $KM+1$ groups connected in a complete graph by global ports. It is the largest diameter three Dragonfly that can be made using the drawers and routers specified. The performance of this Dragonfly is studied in \cite{kim:dragonfly}, \cite{besta:slimfly}, \cite{garcia:routing}.

The relations between $D3(K,M)$ and $MDF(K,M)$ are provided below:

\begin{table}[h!]
	\begin{center}
		\caption{$D3(K,M)\ \  MDF(K,M)$ comparisons}
		\label{tab:table1}

		\begin{tabular}{|l|c|c|} 
			\hline
			\textbf{Property} & \textbf{$D3(K,M)$} & \textbf{$MDF(K,M)$}\\
			\hline \hline
			1 Drawers & $KM$ & $KM+1$\\
			\hline
			2 Fixed Points& Yes &  No\\
			\hline
			3 Scales& Gracefully & Gracefully\\
			\hline
			4 Parallelism over local and global ports & Yes & Not always possible\\
			\hline
			5 $L < K \ \& \ N < M$ & $D3(L,N) \subset D3(K,M)$ & $D3(L,N) \subset D3(K,N)$ \\
			\hline
			6 Global Connectivity & Determined & Many alternatives\\
			\hline
			7 Vector-source Routing & Can be used for parallelism & Not always possible\\
			\hline
		\end{tabular}
	\end{center}
\end{table}

Items $4, 5,$ and $7$ require explanation. The connectivity of $MDF(K,M)$ may be determined by a $KM+1\times KM$ table. The rows are indexed by the drawers and the columns are indexed by (routers, ports). The entries of the table are drawers. The entry at $(d,(p,\gamma))$ is the drawer $d'$ reached by port $\gamma$ of $(d,p)$. In row $d'$, $d$ appears and at the the top of its column is $(p',\gamma ')$ which means the link from $(d,p,(\gamma))$ connects to $(d',p',(\gamma '))$. The only restraint on the table is that row $d$ contain every drawer $d'$ except itself.There is only one table which produces an $MDF$ having property 5 and it is the only way that the Dragonflies are being built. In order for source vector routing to be possible on a Dragonfly, global ports have to permute the set of drawers. The table being used to design machines has the property that each global port maps all groups to only two groups. So no existing Dragonfly supports source-vector routing.

Here are properties common to Swapped Dragonflies and maximal Dragonflies.
\begin{enumerate}
	\item Any pair of routers can be connected by a path containing at most one global hop.
	\item Given routers with radix $r$, the radix devoted to global/versus local can be adjusted in response to the relative cost of long versus short connections.
	\item Deflection routing can be used to avoid conflicts, thereby reducing average latency.
	\item The system UGAL can be used to obtain information about queue lengths that would inform the choice of intermediate routers in a deflection path.
\end{enumerate}

$D3(K,M)$ is a Dragonfly. Any issue associated with long global links is no different than it is for $MDF(K,M)$. On random traffic,the Swapped Dragonfly has the same performance as any Dragonfly of (approximately) the same size. In \cite{kim:dragonfly} results are reported for a Dragonfly with $M=8, K=4$ and $4$ terminals (compute nodes) attached to each router. The network $D3(K,M)$ is of size $256$ and the studies presented apply. In the paper \cite{besta:slimfly}, numerous empirical results are given for the Dragonfly architecture. The parameters are not stated. The network with $N_r = 1452$ in \S 5 has $K=M=11$. $D3(11,11)$ has $1331$ routers. Results given are pertinent to $D3(11,11)$. The paper \cite{garcia:routing} is pertinent to $D3(16,8)$.

\section{RELATED WORK}

There is a vast literature focused on interconnection networks for supercomputers. However, it appears that most of 
it is focused on networks that scale geometrically  rather than linearly. This came about because much of the 
original work focused on networks that implemented algorithms which scaled geometrically.  The research   
was stimulated by the appearance of the CM1  \cite{hill:CM1}, a highly parallel machine with a hypercube  interconnection network. 
At that time, the high degree of the nodes of a hypercube was a problem. This stimulated work such as the 
cube-connected cycle of \cite{prp:CCC}. The cube-connected cycle opened the nodes of a hypercube with a cycle reducing the 
degree of the nodes to three. Had the authors of \cite{prp:CCC}, Preparata and Vuillemin, opened the nodes with a complete graph, they would have invented a Dragonfly in 1981.Other work compared performance on the hypercube to performance on butterfly and other networks. A compendium 
of such results appears in Leighton \cite{lei:book}.  

Much of the research on interconnection networks (a.k.a., topology) was done during the 80's and 90's. 
Since that time, it appears to this author that there has been much less work on supercomputer networks. However, two bodies of work are closely related to the work in this paper. Each start with a building block graph which is extended by some means. This building block corresponds to a drawer.

The first body of work \cite{mar:OTIS} dates to 1993. 
A graph $(V,E)$ of order $n$ is the building block. 
A new graph of order $n^2$, called an \emph{OTIS network}, is defined using the swap $(d,p) \leftrightarrow (p,d)$. 
Subsequent work on these networks was done in
\cite{ost:sort,rarj:OTISsort,xao:biswap}.   
The use of the swap  was discovered, apparently independently, in 1996 \cite{yeh:unify}   where the resulting network 
is called a \emph{swapped network}.
The swapped network with nucleus $G$ is our single cabinet network $D3(1,M)$ provided that the building block $G$ 
is a complete graph of order $M$. The paper \cite{yeh:recurswap}  uses the swap recursively to create a hierarchy of networks. 
The orders of the graphs in the  resulting hierarchy grow geometrically rather than linearly. 
Additional work on swapped networks appears in \cite{yeh:unify, par1:swapprop, par2:swapdef,yeh:recurswap}. 

There is disagreement over whose use of the swap takes precedence\footnote{See Section 3 of \cite{par1:swapprop} and \cite{par2:swapdef}.}. Certainly this work does not. This contribution is the addition of the cabinet coordinate $c$ which creates linear scaling and constant diameter for all members of the family.

Swapped networks are not vertex transitive because the node $(p,p)$ is left fixed by the swap. Disappointed with this fact, the authors of \cite{xao:biswap} defined bi-swapped networks. A bi-swapped  network also is defined based upon a nucleus graph $G$. It is a subgraph of the two cabinet graph  $D3(2,M)$ if the nucleus of the bi-swapped graph  is a complete graph of order $M$. 
The authors chose to expand hierarchically.  Had they expanded linearly, they would have defined $D3$ graphs. Their paper  compares bi-swapped networks to swapped and OTIS networks.

The second body of work was motivated by the availability of inexpensive high density electrical interconnect and expensive but fast optical interconnect.The first paper in this work appeared in 2002 \cite{gup:SOENet}. The resulting network is called a Scalable Opto-Electronic Network (SOENet). It, too, is a linearly\footnote{The author refers to it as economically scalable.} scalable family of  networks built from uniform multi-router building blocks.  The authors did not focus on a low diameter network but on linear scalability and cost. A SOENet connects the subnetworks to a switching fabric with optical connections. The resulting system scales linearly but diameter is a step function.

There does not seem to be much follow-up on the name SOENet in the literature. However, it is referenced in and almost surely influenced the next paper \cite{kim:dragonfly} which appeared in 2008. This paper  does not employ the swap but does extend one graph by another, including the 
case that both are complete graphs.\footnote{See section 3.10 of \cite{kim:dragonfly}} The resulting graph is called a \emph{ Dragonfly}. Analogous to the research that followed the CM1, the authors of the Dragonfly were making a collection of routers in a group behave like a very high radix router. A group is the same as a drawer and the collection of global ports is joined by the local network to simulate a global router with $MK$ ports.

The Dragonfly is clearly the most important of the networks discussed here. The IBM PERC \cite{Arim:Inter} system uses a Dragonfly topology. It is the network of the $\text{CRAY XC}^\copyright$ according to \cite{alv:CRAY}. Both the $\text{CRAY XC}^\copyright$ and $D3(K,M)$ scale  linearly up to a technology determined bound. 
The CRAY terminology  differs from the terminology in this paper and from \cite{kim:dragonfly}. A chassis has sixteen routers 
connected electrically as a complete graph. Three  chassis are  electrically wired in a group, which physically occupies two cabinets and is diameter two. 
Groups are then connected with fiber optic cable. The result has diameter at least five. Very large machines have diameter greater than five.  
The correspondence between the terms in this paper and CRAY's terms is imperfect. A chassis corresponds to a drawer. A group corresponds to a cabinet. The optical or global network corresponds to a global network except for one thing. The drawer to drawer connections on a cabinet are part of the global network whereas they are not part of the optical network in the CRAY terminology. If $D3(K,M)$ were defined on a drawer which was of diameter two, the 
correspondence would be closer. Although the interconnection network has strong similarities to $D3(K,M)$, the implementation 
in the $\text{CRAY XC}^\copyright$ employs a form of Valiant randomized routing to reduce congestion. Therefore, there is no analysis of congestion free all-to-all exchange.

A network called Slimfly is a diameter two network, which is highly desirable. It is designed using an idea exploited by graph theorists to approach the Moore Bound for graphs of  given degree and node count \cite{McKay:Moore Bound}. Slimfly uses a finite field in its construction which determines the size of the graph. Finite fields have size $p^k$ for primes $p$ and powers $k$ so there are lots of choices for size. However, there is no linear scalable family of Slimflies which is a disadvantage. 

\section{CONCLUSION}

The Swapped Dragonfly is a linearly scalable family of networks. The networks are built form identical subnets and have diameter three. There is an intuitive as well as technical description of the global network connections. $D3(K,M) \subset D3(J,N)$ if $J\leqq K$ and $L \leqq M$. This has several ramifications. A $D3(J,M)$ can be expanded to a $D3(K,M)$ without disturbing the wiring of $D3(J,M)$. A drawer can be removed for repair leaving a $D3(K,M-1)$ running. The Swapped Dragonfly is capable of parallelism over local ports, global ports, routers, and vector paths. This property distinguishes the Swapped Dragonfly from the Dragonfly. 

This paper has studied $D3(K,M)$ as a packet switching network. In order to accelerate communication primitives, four ideas are used: source-vector routing, pipelining, synchronization, and the swap. None of these ideas are original but the use here of the swap is original. Because of the swap, if every router of a drawer sends a packet $(3;\gamma,\pi,\delta)$ there are no link conflicts. The paths are "parallel" in the sense that they do not cause link conflicts. They do have to be simultaneous and the synchronization counter is used to ensure that property. Finally, pipelining is used to compress the time to run an entire algorithm. There are no inter-round conflicts  and intra-round conflicts are resolved by a one round delay in the source program. The result is the timings for communication primitives displayed below:

\begin{enumerate}
	\item $(c,d,p)$:$\ \ $ $N$ broadcasts in $N$ rounds.
	\item $(c,p,p)$:$\ \ $ $N$ broadcasts in $2N$ rounds.
	\item $(c,d,p)$:$\ \ $ one-to-all in $\min(M^2, KM)$ rounds.
	\item $(c,p,p)$:$\ \ $ one-to-all in $KM$ rounds with $M$ intra-round conflicts.
	\item $(c,d,p)$:$\ \ $ all-to-one in $KM$ rounds.
	\item $D3(K,M)$:$\ \ $ all-to-all in $KM^2$ rounds with $KM$ intra-round conflicts.
	\item Permutation in $M$ network hops.
\end{enumerate}

The algorithms producing these results are presented and fully analyzed in the appendix. They use parallelism over local ports, global ports, and all routers. The paper shows that source-vector routing makes designing efficient parallel algorithms easy. Source-vector routing also distinguishes the Swapped Dragonfly from the Dragonfly.

There are ways in which this model of network routing could be implemented. It would be necessary to have two levels of service, standard and privileged. Privileged packets would take precedence and could not be impeded. 
It is also possible to use destination routing. This makes forms of deflection routing and systems like UGAL possible. The communication primitives can be converted to destination routing. Rounds will be free of link conflicts but the analysis of intra-round conflicts will not apply. 

There are other results which are developed in papers that are in preparation. If $K$ and $M$ have a common factor $S$, then there is a doubly-parallel algorithm that performs an all-to-all in $KM^2/S$ rounds with $KM$ inter-round conflicts \cite{dra:all-to-all}. If $K=L^2$, a vector-matrix product can be performed in four network hops and a matrix product can be performed in $KM$ rounds \cite{dra:vecmat}. If $K=2^k$ and $M=2^m$, $D3(K,M)$ is a diameter three wiring of the $(k+2m)$-dimension binary hypercube. It contains a dilation two simulation of the hypercube which makes ascend-descend algorithms possible \cite{dra:hyper}.

\section{Acknowledgments}
The author would like to thank Bill Carlson, Duncan Roweth, and Patricia Draper for many helpful discussions, and Robert Mroskey of the Laboratory for Physical Sciences in Catonsville has simulated $D3(K,M)$. Research for this paper was supported by The Center for Computing Sciences.

{\raggedright

}

\newpage

\begin{center}
	\normalsize\bfseries\MakeUppercase{APPENDIX}
\end{center}

This Appendix contains proofs of the claims made in \S9 and three figures demonstrating the structure of $D3(3,4)$. Each proof is an algorithm to be run by a compute node attached to each router.  The algorithm consists of a series of packet headers to be launched by the node. Sometimes it is necessary to insert a delay to avoid link conflicts among packets of the algorithm. The delays are the responsibility of the attached node. They are denoted in the algorithms and proofs by a "false" packet header with $b \ge 1$ and vector $(0,0,0)$. In displays showing the evolution of an algorithm the delays are indicated by repeating the router address $(c,d,p)$. Proofs consist of verifying that no link conflicts occur.

The performance of $D3(K,M)$ on communication primitives: broadcast, one-to-all, all-to-one, all-to-all, and permutation is determined. Throughout this section $KM^2$ is even. The analysis assumes that the network is balanced as described in $\S 7$. Using synchronizing headers, the first four primitives can be carried out without any link conflicts\footnote{The discussions in this section and $\S 9$ are stated in terms of $D3(K,M)$. They apply to $D3(\kappa,M,N)$ via the isomorphism in Theorem 1. Where arguments below use the ordering of the parameter $c$, the ordering of $k_i$ would be based upon $i$.}. Headers have the form $(B,b;\gamma,\pi,\delta)$. $B$ is the broadcast bit and $b$ is the synchronizing bit. 

There are several protocols for which pipelining of communication is used. Each is implemented by a list of headers, which are applied in order. The list is then iterated with entries in the headers being modified.

\begin{equation*}\label{eq:proc1}
\begin{array}{c}
\textrm{Protocol $ 1$}\\
\begin{array}{cccccccc}
\textrm{step} & 1 & 2 & 3 & 4 & 5 \\
& l & g & l &&&& \\
&   & l & g & l &&& \\
\textrm{rnds}    &   &   & l & g & l && 
\end{array}
\end{array}
\end{equation*}

\begin{equation*}\label{eq:prot2}
\begin{array}{c}
\textrm{Protocol $ 2$}\\
\begin{array}{cccccccccc}
\textrm{step} & 1 & 2 & 3 & 4 & 5 & 6 & 7 \\
& l & g & l &  &&&&&\\
&   & l & g & l &  &&&&\\
\textrm{rnds}&   &   &   & l & g & l &&&\\
&   &   &   &   & l & g & l &&
\end{array}
\end{array}
\end{equation*}

\begin{equation*}\label{eq:prot3}
\begin{array}{c}
\textrm{Protocol $ 3$}\\
\begin{array}{cccccccccc}
\textrm{step} & 1 & 2 & 3 & 4 & 5 & 6 & 7 & 8 & 9\\
& l & g & l &  &&&&&\\
&   & l & g & l &  &&&&\\
\textrm{rnds}&   &   &   && l & g & l &&\\
&   &   &   &  && l & g & l &
\end{array}
\end{array}
\end{equation*}

The first protocol is a unit cost protocol; $N$ rounds with no delays. The second is a $3/2$ unit cost protocol; $N$ rounds with $N/2$ delays. The third is a $2$ unit cost protocol; $N$ rounds with $N$ delays. It is clear that there can be no intra-round link conflicts in Protocol $3$ because the global and local networks are independent. The second Protocol leads to a local link conflict at step $4$ if the local ports at rounds $2$ and $3$ are the same. If conflicts are resolved by last-in first-out the the conflict causes no delay to the algorithm.

\setcounter{theorem}{3}
\begin{theorem}
	A router $(c,d,p)$ of $D3(K,M)$ can perform a broadcast in three hops. It can perform N broadcasts in $N$ rounds if $d \neq p$ and in $2N$ rounds if $d = p$.
\end{theorem}
\proof
The algorithm for a broadcast is the single instruction $(c,d,p)(1,3;0,0,0)$. The first bit is the broadcast bit.
The evolution of the algorithm is:
\begin{equation*}
\begin{array}{rcc}
&router &  header\\
&(c,d,p)& (1,3;0,0,0)\\
&(c ,d,*)&(1,2;0,0,0)\\
&(*,*,d)&(1,1;0,0,0)\\ 
&(*,*,*) &(1,0;0,0,0)\\
\end{array}
\end{equation*}
At the end of three hops, this sequence delivers $(c,d,p)$'s message exactly once to every router\footnote{The message arrives at some routers before $b=0$. However, they do not read the message until $b=0$. This is critical for the all-to-one algorithm}.

This algorithm can be chained using Protocol 1. Assume $(c,d,p)$ has N messages to broadcast and $p \neq d$. 
The chained algorithm is:
\begin{align*}
&\textrm{For $i = 1,\dots,N$,}\\
&\phantom{MM}(c,d,p)(1,3;0,0,0) \phantom{MM}\textrm{message $i$}\\
&\textrm{end for.}
\end{align*}
Without displaying the header, the effect of this loop is:
\begin{equation*}
\begin{array}{lcccc}
&&&\ \textrm{rnd} &\\
&1 &(c,d,p) && \\
&2 &(c,d,*) &(c,d,p) &\\
\textrm{Step}  &3 &(*,*,d) &(c,d,*)  &(c,d,p) \\
&4 &(*,*,*) &(*,*,d) &(c,d,*) \\ 
&5 & &(*,*,*) &(*,*,d) \\
&6 & & &(*,*,*) \\
\end{array}
\end{equation*}

At step 2, rounds $1$ and $2$ use global and local networks, respectively, so there is no conflict. At step $3$, rounds $1$ and $3$ conflict if $d = p$.
Therefore it is necessary to use Protocol $3$, which is implemented by the following algorithm:
\begin{align*}
&\textrm{For $i = 1,3,5,\dots,2n + 1$}\\
&\phantom{MM}(c,p,p)(1,3;0,0,0) \phantom{MM}\textrm{message $i$}  \\
&\phantom{MM}(c,p,p)(1,3;0,0,0) \phantom{MM}\textrm{message $i + 1$} \\
&\phantom{MM}(c,p,p)(0,1;0,0,0) \phantom{MM}\textrm{no message} \\
&\phantom{MM}(c,p,p)(0,1;0,0,0) \phantom{MM}\textrm{no message} \\ 
&\textrm{end for}
\end{align*}

The evolution of this algorithm, without showing headers is:

\begin{equation*}
\begin{array}{lcccccc}
& 1 &(c,p,p) &&&&\\
& 2 &(c,p,*) &(c,p,p) &&&\\
& 3 &(*,*,p) &(c,p,*)  &(c,p,p) &&\\
\textrm{Step} & 4 &(*,*,*) &(*,*,p) &(c,p,p) &(c,p,p) &\\ 
& 5 & &(*,*,*) &(c,p,p) &(c,p,p) & (c,p,p)\\
& 6 &&& (c,p,p) &(c,p,p) &(c,p,*)
\end{array}
\end{equation*}
There is no conflict at step $5$ because round $2$ is receiving while round $3$ is sending. Protocol $3$ delivers $N$ broadcasts from $(c,d,p)$ in $N$ rounds with $N$ intra-round delays.
\qed
Routers cannot duplicate values so each round of the broadcast involves stepping off and on the network. This synchronizes the rounds.

In the next algorithm, the statement $(c,d,p)(0,3;\gamma,\pi,*)$ means that $(c,d,p)$simultaneously sends $(0,3;\gamma,\pi,\delta)$ for all $\delta$. The expression $Lgl$ means that the first step is parallel.

\begin{theorem}
	A one-to-all communication can be performed in $KM$ rounds if $p\neq d$ and in $KM$ rounds with $M$ intra-round delays if $p= d$.
\end{theorem}

\proof

\textrm{Let $i = \pi + \gamma M$}
\begin{align*}
&\textrm{For $i = 1,\dots,KM-1$}\\
&\ \ \ \ \ \ \ (c,d,p)(0.3;\gamma,\pi,*)\\
&\textrm{end for}
\end{align*}

At the first step of a round $M$ messages are scattered to the $M$ routers of drawer $d$. At the second step they move to column $d$ of cabinet $c+\gamma$. The algorithm evolves as 
\begin{equation} 
\begin{array}{crclclclcl}
rnd & 1 &2& 3& 4 & 5&step & \\
i & L & g & l \\
i+1 && L & g & l \\ 
i+2 &&& L & g &  l.
\end{array}
\end{equation} 
At step $2$ of round $1$ the drawer is sent to a column so there are no conflicts within a round.
At step $3$ rounds $i$ and $i+2$ are at

$$
\begin{array}{cclclc}
&i&(c+\gamma,*,d)& \stackrel{\pi}\longrightarrow&(c+\gamma,*,d+\pi)\\
&i+2&(c,d,p)& \stackrel{*}\longrightarrow&(c,d,*).
\end{array}
$$
If $d \neq p$, there is no conflict. If $d=p$ and $\gamma \neq 0$, there is no conflict. But if $d=p$ and $\gamma = 0$, link $\pi$ is carrying two packets. In that case a delay is necessary. This happens for $M$ values of $\pi$. The algorithm has to be modified appropriately.
\qed

An all-to-one communication is performed by having the sink node send a sequence of broadcasts requesting responses from $M$ nodes at a time. The protocol is:
\begin{equation*}
\begin{array}{llllllllllll}
& L& G& L& D& l& g& l&&&&\\
&& L& G& L& D& l& g& l&&
\end{array}
\end{equation*}

where capital letters represent broadcasts and $D$ is a delay.

\begin{theorem}
	The sink $(c,d,p)$ \  can perform an all-to-one communication in $KM$ rounds if $d\neq p$.
\end{theorem}

\proof \textrm{Let $i=\pi+\gamma M$}
\begin{align*}
&\textrm{For $0 \leq i < KM$}\\
&\phantom{MMMM}(c,d,p)(1,3;0,0,0 ) \ \textrm{message}(0,3;\ \gamma,\delta,\pi)\\
&\textrm{end for}.
\end{align*}
Receiving node $(c',d',p')$ interprets the message as the request to send message to $(c,d,p)$ if $c'=\gamma \textrm{ and } p'=\pi$
$$(\gamma,d',\pi)(0,3;c-\gamma,p-\delta,d-\pi).$$ The value of $(\gamma,\pi)$ is determined by $i$. The result is that $M$ nodes send their message to $(c,d,p)$ at each step of $i$.
The evolution of the algorithm is:

\begin{equation*}
\begin{array}{clcccc}
&rnd& i&i+1&i+2\\ 
step&&&&&\\
&&(c,d,p) &&&\\
1& & L\downarrow&&&\\
&&(c,d,*) & (c,d,p)&&\\
2& & G\downarrow&L\downarrow&&\\
&&(*,*,d) &(c,d,*) & (c,d,p)&\\
3& & L\downarrow& G\downarrow&L\downarrow&\\
&&(*,*,*) &(*,*,d) &(c,d,*)&\\
4& & D\downarrow& L\downarrow&G\downarrow&\\
&&(\gamma,*,\pi) &(*,*,*) &(*,*,d) &\\
5& & l\downarrow& D\downarrow& L\downarrow&\\
&&(\gamma,*,d) & (\gamma',*,\pi') &(*,*,*)&\\
6& & g\downarrow& l\downarrow& D\downarrow&\\
&&(c,d,*) & (\gamma',*,d)& (\gamma ",*,\pi")&\\
7&& l\downarrow& g\downarrow&l\downarrow&\\
&&(c,d,p)&(c,d,*) &(\gamma",*,d)&\\
8& && l\downarrow& g\downarrow&\\
&&&(c,d,p)&(c,d,*)\\
9&&&&l\downarrow\\
&&&&(c,d,p).
\end{array}
\end{equation*}

Each column is doing $LGLDlgl$. The delay $D$ allows the attached node to launch the return packet if $(c',p') = (\gamma,\pi)$. There is no conflict at step $3$ because $d \neq p$. There is no conflict at step $5$ because if $\pi = d$, step $5$ of round $i$ is a router delay. There is no conflict at step $6$ because if $\pi' = d$, step $6$ of round $i+1$ is a router delay. There is no conflict at step $7$ because $d \neq p$. The sink $(c,d,p)$ receives $M$ messages every step after step $6$ so the algorithm takes time $KM+6$.
\qed

An analogous algorithm parallel over global ports performs the one-to-all in $M^2$ rounds if $d \ne p$. It uses $glgl$ paths which jump to the destination cabinet on the first hop.

In the next algorithm the statement $(*,*,*)(0,3;\gamma,\pi,\delta)$ means that all $(c,d,p)$ simultaneously send a message with packet header $(0,3;\gamma,\pi,\delta)$. The statement $(*,*,*)(0,1;0,0,0)$ means that a one hop delay is taken.

\begin{theorem}
Assume  $M \geq 4$. An all-to-all exchange can be performed in $KM^2$ rounds with $KM$ intra-round delays.
\end{theorem}

\proof
\begin{alignat*}{4}
\textrm{Let }i=\pi&+\delta M+\gamma M^2&&\\
&\textrm{For }i=0,\cdots,KM^2-1\\
&\phantom{MM}\textrm{test }=\pi(i)-2 \bmod M\\
&\phantom{MM}\textrm{if test } =\delta(i)\\
&\phantom{MMMM} (*,*,*)(0,1;0,0,0)\\
&\phantom{MMMM}(*,*,*)(0,3;\gamma(i),\pi(i),\delta(i))&&&\\
&\phantom{MM}\textrm{else}\\
&\phantom{MMMM}(*,*,*)(0,3;\gamma(i),\pi(i),\delta(i))&&&\\
&\phantom{MM}\textrm{end if}&&&\\
&\textrm{end for}&&&\\
\end{alignat*}

This algorithm implements Protocol 1 simultaneously at every router of $D3(K,M)$. 
Each router $(c,d,p)$ sends its message $i$ at round $i$ on the path 
\begin{equation*}
\begin{array}{rcc}
&router &  header\\
& (c ,  d, p)&(0,3;\gamma,\pi,\delta)\\
l_1 & \downarrow & \downarrow\\
&(c ,d,p+\delta)&(0,2; \gamma ,\pi, \delta)\\
g & \downarrow &\downarrow\\
&(c + \gamma  ,p+\delta,d)&(0,1;\gamma, \pi),\delta\\ 
l_2 & \downarrow & \downarrow\\
&(c + \gamma  ,p+\delta,d+\pi)&(0,0;\gamma,\pi,\delta).
\end{array}
\end{equation*}

There are no link conflicts within round $i$ because router  $(c,d,p)$ communicates with drawer $(c+\gamma,p+\delta)$ and this is in a different drawer for each $p$. There are no link conflicts between round $i$ and round $i+1$ because local and global links can be traversed independently and simultaneously.
However, there can be a link conflict between round $i$ and round $i+2$. In Protocol $1$, round $i$ and $i+2$ both use a local port at the same time. If $\pi(i) =\delta(i+2)$, a local link conflict occurs because all $(c,d,p)$ are acting in unison. The test in the algorithm imposes a one-step delay on round $i+2$ which prevents the conflict. This delay occurs $KM$ times.
\qed

The one-to-all and all-to-one algorithms are presented using Protocol $1$ interrupted by delays when link conflicts occur. Each can use Protocol $2$ or $3$ to eliminate link conflicts. These Protocols increase the total hop count for the all-to-all algorithm to $3KM^2/2$ and $2KM^2$, respectively. 

A \emph{permutation}, $\pi$, is a collection of $KM^2$ messages for which 
$$
(c,d,p) \rightarrow \pi(c,d,p) =(e,f,d)
$$ 
is a permutation of the routers of $D3(K,M)$.

\begin{theorem} 
A permutation can be performed in less than or equal to $M+4$ hops.
\end{theorem} 
\proof
The nodes attached to routers in a drawer send their destination to all other routers in the drawer in one hop. Three hop paths only conflict if they are drawer to drawer. If only one pair of drawers is involved $glgl$ paths prevent conflict. But if more than one pair of drawers is involved, $glgl$ paths from one pair can conflict with $glgl$ paths from another pair. Because this is a permutation, at most $M$ such conflicts can occur. Therefore it is better to use $glgl$ paths when dealing with drawer to drawer conflicts. If the $glgl$ paths do not interfere with each other, the permutation takes only $5$ hops, if they do interfere with each other the permutation may take $M+4$ hops.
\qed

\begin{center}
	\begin{flushleft}
		\begin{figure*}
			\includegraphics{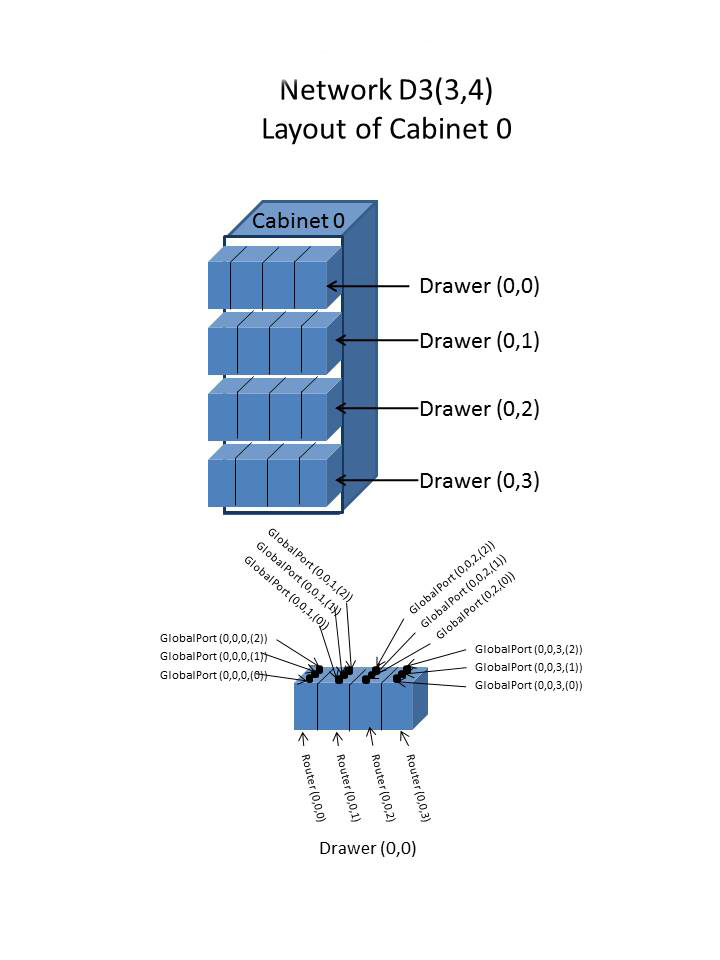}
		\end{figure*}
	\end{flushleft}
\end{center}
	
\begin{center}
	\begin{flushleft}
		\begin{figure*}
			\includegraphics{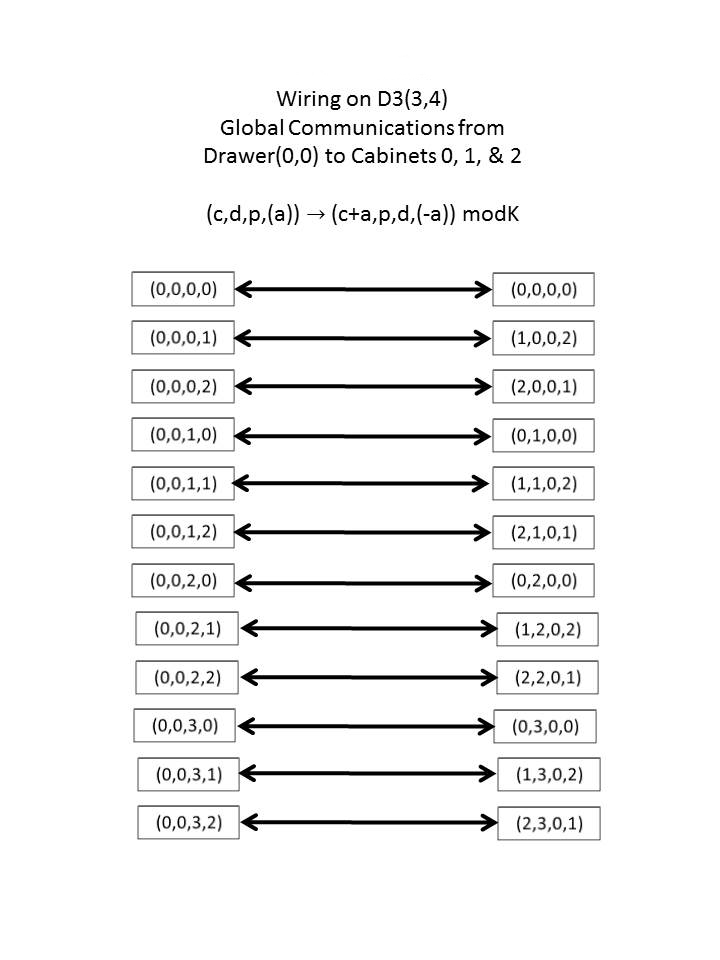}
		\end{figure*}
	\end{flushleft}	
\end{center}

\begin{center}
		\begin{figure*}
\hspace{1.4in}
		\includegraphics{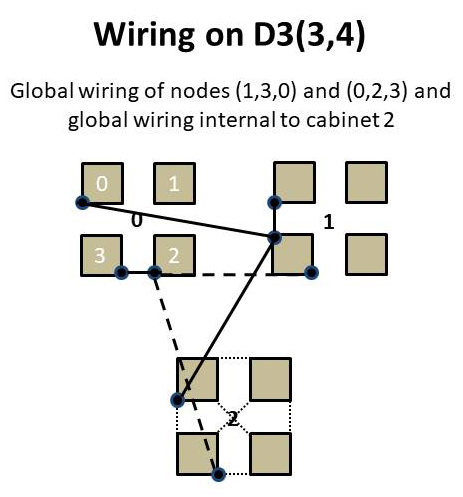}
		\end{figure*}
\end{center}
\end{document}